\documentclass[11pt,pdftex]{article}

\usepackage{amssymb}
\usepackage{amsmath}
\usepackage{epigraph}
\usepackage{color} 
\usepackage{fullpage}
\usepackage{graphics}
\usepackage{url}
\usepackage{cite}
\usepackage{multicol}

\newtheorem{theorem}{Theorem}
\newtheorem{definition}{Definition}
\newenvironment{proof}[1][Proof]{\noindent\textbf{#1.} }{\hfill $\Box$\\[2mm]} 
\def\lf{\tiny}

\def\nnll{\refstepcounter{linenumber}\lf\thelinenumber}

\newcommand{\myparagraph}[1]{\vspace{2pt}\noindent \textbf{#1}}

\newcounter{linenumber}

\def\S{\ensuremath{\mathcal{S}}}

\def\A{\ensuremath{\mathcal{A}}}

\def\I{\ensuremath{\mathcal{I}}}
\def\O{\ensuremath{\mathcal{O}}}

\def\T{\ensuremath{\mathcal{T}}}

\def\Nat{\ensuremath{\mathbb{N}}}

\def\val{\textit{val}}

\newcommand{\remove}[1]{}

\newcommand{\id}[1]{\mbox{\textit{#1}}}

\newcommand{\ignore}[1]{}

\def\AL{\textit{AL}}
\def\SCN{\textit{SCN}}

\begin{document}

\title{Set-Consensus Collections are Decidable}

\author{
Carole Delporte-Gallet
\protect\footnote{IRIF, Universit\'e Paris-Diderot} 
\hspace{1cm}
Hugues Fauconnier
\protect\footnotemark[1]
\\
\\
Eli Gafni
\protect\footnote{UCLA} 
\hspace{1cm}
Petr Kuznetsov
\protect\footnote{T\'el\'ecom ParisTech, corresponding author,
  \texttt{petr.kuznetsov@telecom-paristech.fr}}
\thanks{The research leading to these results has
  received funding from the Agence Nationale de la Recherche,  under
  grant agreement ANR-14-CE35-0010-01, project DISCMAT.}  
}
\date{}
\maketitle

\begin{abstract}
A natural way to measure the power of a distributed-computing model
is to characterize the set of tasks that can be solved in it. 
In general, however, the question of whether a given task can be
solved in a given model 
is undecidable, even if we only consider the wait-free shared-memory model. 

In this paper, we address this question for restricted classes  of models and tasks. 
We show that the question of whether a collection $C$ of
\emph{$(\ell,j)$-set consensus} objects, for various $\ell$ (the number of
processes that can invoke the object) and $j$ (the number of distinct
outputs the object returns), can be used by $n$ processes
to solve wait-free $k$-set consensus is decidable. 
Moreover, we provide a simple $O(n^2)$ decision algorithm, 
based on a dynamic programming solution to the Knapsack optimization problem.

We then present an \emph{adaptive} wait-free set-consensus algorithm that,
for each set of participating processes, achieves the
best level of agreement that is possible to achieve using $C$.
Overall, this gives us a complete characterization of a read-write
model  defined by a collection of set-consensus objects through its
\emph{set-consensus power}.

We conjecture that any ``reasonable'' shared-memory can be represented
by a collection of set-consensus tasks and, thus,
characterized by the set-consensus power.
\end{abstract}





\section{Introduction}
A plethora of models of computation were proposed for distributed
environments. The models vary in timing assumptions they make, types of
failures they assume, and communication primitives they employ. 
It is hard to say \emph{a priori} whether one model provides more
power to the programmer than the other.  
A natural way to measure this power 
is to characterize the set of distributed  tasks that can be solved in
a model. 
In general, however, the question of whether a given task can be
solved in the popular \emph{wait-free} read-write model, i.e., tolerating asynchrony and failures of
arbitrary subsets of processes, is
undecidable~\cite{GK99-undecidable}.
Of course, in models in which processes can additionally access arbitrary objects, the question is not
decidable either.
However, many natural models have been shown to be characterized by their
power to solve set consensus~\cite{DFGT11}.

In this paper, we consider models in which $n$ completely asynchronous processes
communicate through reads and writes in the shared memory and, in addition,
can access \emph{set-consensus} objects.   
An $(\ell,j)$-set-consensus object solves $j$-set consensus among $\ell$
processes, i.e., the object can be accessed by
up to $\ell$ processes with \emph{propose} operations that
take natural numbers as inputs and return natural numbers as
outputs, so that the set of outputs is a subset of inputs of size at
most $j$.
Set consensus is a generalization of consensus and, like consensus~\cite{Her91},
exhibits a \emph{universailty} property: $\ell$ processes can use
$(\ell,j)$-set consensus and read-write registers to implement $j$
state machines, ensuring that at least one of them makes
progress~\cite{GG11-univ}.
In this paper, we explore what level of agreement, and thus ``degree of
universality'', can be achieved using any number of objects from a
given set-consensus collection.       

The special case when only one type of set consensus
can be used in the implementation was resolved
in~\cite{BG93-TR,ChaR96,Rei96}.
Assuming that $k\geq j\lceil n/\ell \rceil$, we trivially solve
$j\lceil n/\ell \rceil$-set consensus, by splitting $n$ processes into $\lceil n/\ell \rceil$ groups
of size $\ell$ (or less).
A slightly more complex converse bound~\cite{BG93-TR,ChaR96,Rei96}, accounting for the
``delta'' between $n$ and $\ell\lceil n/\ell \rceil$, resolves the
special case when only one type of set consensus object can be used.


Characterizing a general model in which processes communicate via
 objects in an arbitrary collection $C$ of possibly different
 set-consensus objects is more difficult.
For example, let $C$ be $\{(2,1),(5,2)\}$, i.e.,
every $2$ processes in our system can solve consensus and every
$5$ can solve $2$-set consensus. 
What is the best level of agreement we can achieve using registers and
an arbitrary number of objects in $C$ in a system of $9$
processes? 
One can easily see that $4$-set consensus can be solved: 
the first two
pairs of processes solve consensus and the remaining $5$ invoke
$2$-set consensus, 
which would give at most $4$ different outputs.
One can also let the groups of the first $5$ and the remaining $4$
each solve $2$-set consensus. 
(In general, any two set-consensus objects $(\ell_1,j_1)$ and $(\ell_2,j_2)$ can be
used to solve $(\ell_1+\ell_2,j_1+j_2)$-set consensus.) 
But could we do $(9,3)$-set consensus with $C$?
   
We propose a simple way to characterize the power of a
set-consensus collection.  
By convention, let $(\ell_0,j_0)$ be $(1,1)$, and note that $(1,1)$-set
consensus is trivially solvable.
We show that a collection $C=\{(\ell_0,j_0),
(\ell_1,j_1),\ldots,(\ell_m,j_m)\}$ solves $(n,k)$-set consensus
if and only if there exist
$x_0,x_1,\ldots,x_m\in \Nat$, such that
$\sum_i \ell_i x_i \geq n$ and $\sum_i
j_i x_i \leq k$. 
Thus, determining the power of $C$ is 
equivalent to solving a variation of  the Knapsack optimization
problem~\cite{knapsack}, where each $j_i$ serves as the ``weight''
of an element in $C$, i.e., how much disagreement it may incur, and each
$\ell_i$ serves as its ``value'', i.e., how many processes it is able to
synchronize.
We describe a simple $O(n^2)$ algorithm for computing the power of $C$ for
solving set consensus among $n$ processes using the dynamic
programming approach.

The sufficiency of the condition is immediate.
Indeed, the condition implies that we can partition the set of $n$
processes in $\sum_i x_i$ groups:
$x_0$ groups of size (at most) $\ell_0$, $x_1$ groups of size (at most) $\ell_1$,
$\ldots$,
$x_m$ groups of size (at most) $\ell_m$.  
Each of the $x_i$ groups of size $\ell_i$, $i=0,\ldots,m$, can independently solve
$j_i$-set consensus using a distinct $(\ell_i,j_i)$-set-consensus
object in $C$, which gives us at most
$\sum_i j_i x_i \leq k$ different outputs in total.     

The necessity 
uses a generalized version of the BG
simulation~\cite{BG93b,BGLR01} that allows to simulate, in the
read-write shared-memory model,  a protocol that uses various types of
set-consensus objects.  
We use this simulation to show that if a collection not satisfying the condition
solves $(n,k)$-set consensus, then $k+1$ processes can solve $k$-set
consensus using read-write registers, contradicting the classical wait-free set-consensus
impossibility result~\cite{BG93b,HS99,SZ00}.        
Interestingly, the necessity of this condition holds even if we can use
read-write registers in addition to the elements in $C$.  

Thus, we derive a complete characterization of models defined by 
collections of set-consensus objects. 
In particular, it allows us to determine the \emph{$j$-set-consensus
  number} of a set-consensus collection $C$ as the maximal number of
processes that can achieve $j$-set
consensus using $C$ and read-write registers. 
Applied to arbitrary objects, this metric is a natural generalization of Herlihy's consensus
number~\cite{Her91}. 

Coming back to the collection $C=\{(2,1),(5,2)\}$, 
our characterization implies that $4$ is the best level
of set consensus that can be achieved by $9$ processes with $C$. 
Observe, however, that if only $2$ processes participate, then they can use $C$ to solve
consensus, i.e., to achieve ``perfect'' agreement. 
Applying our condition, we also see that participating sets of sizes $3$ up to $5$
can solve $2$-set consensus, participating sets of size
$6$ up to $7$ can solve $3$-set consensus,
participating sets of size $8$ up to $10$ can solve $4$-set consensus,
etc.
That is, for every given participating set, we can devise an \emph{optimal}
set-consensus algorithm that ensures the best level of
agreement achievable with~$C$. 

An immediate question is whether we could \emph{adapt} to the participation
level and ensure the best possible level of agreement in any case? 
Such algorithms are very useful in large-scale systems with
bounded contention levels.
We show that this is possible by presenting an \emph{optimally
  adaptive} set-consensus algorithm. 
Intuitively, for the currently observed participation, our algorithm 
employs the best algorithm and, in case the participating set grows,
seamlessly relaxes the agreement guarantees by switching to a possibly less
precise algorithm when there is a larger set of participants.

Our results thus imply that there is an efficient algorithm to decide whether one model defined by a collection of set-consensus object types can be implemented in model defined by another collection of set-consensus objects.
%
%
We conjecture that the ability of any ``reasonable''
(yet to be defined precisely)  
shared-memory system to solve set consensus, captured by its
$j$-set-consensus numbers, for all positive $j$,
characterizes precisely its computing power with respect to 
solving tasks or implementing deterministic objects.

This work contributes to the idea that there is
nothing special about consensus that set consensus cannot do. Indeed,
set-consensus collections are decidable in the same way collections
of consensus objects are~\cite{Her91}: the
power of a collection of consensus objects $\{(\ell_1,1),(\ell_2,1),\ldots,(\ell_m,1)\}$  to solve consensus is
determined by $\max_i \ell_i$.
Furthermore, it was recently shown that the computational power of a
class of  deterministic objects cannot be characterized by its ability
to solve consensus~\cite{AEG16}, which suggests the use of set
consensus in a characterization.  
We see this paper as the first step towards proving the conjecture
that the computational power of a deterministic object 
can be captured by its \emph{set-consensus} number,
determining the best level of agreement the object can reach for each
given system size.

\myparagraph{Roadmap.} 
The rest of the paper is organized as follows. 
In Section~\ref{sec:model}, we recall the basic model definitions and
simulation tools.  In  Section~\ref{sec:char}, 
we present and prove our characterization of
set-consensus collections, and describe an efficient
algorithm to compute the characterizing criterion.  
In Section~\ref{sec:adapt}, we present an adaptive algorithm that
achieves the optimal level of agreement for each set of active
participants having access to a given set-consensus collection.
We discuss related work in Section~\ref{sec:related} and conclude in Section~\ref{sec:conc}.

\section{Preliminaries}
\label{sec:model}

In this section, we briefly state our system model, recall the notion
of a distributed task, and sketch the basic simulation tools that we
use in the paper.

\myparagraph{Processes and tasks.}
We consider a system $\Pi$ of asynchronous processes
that communicate via shared memory abstractions.
We assume that process may only fail by crashing, and otherwise it must respect the
algorithm it is given. 
A \emph{correct} process  never crashes.
Shared abstractions we consider here include an \emph{atomic-snapshot}
memory~\cite{AADGMS93} and a collection of objects solving \emph{distributed
  tasks}~\cite{HS99}.

An atomic-snapshot memory stores a vector of $|\Pi|$ values, one
value per process, and exports atomic operations $\textit{update}$ and
$\textit{snapshot}$: operation $\textit{update}(p,v)$ performed by
process $p$  writes $v$ in position $p$ in the vector, and operation
$\textit{snapshot}()$ returns the vector.  
Atomic-snapshot memory can be implemented, in a wait-free and linearizable manner,  
in the standard read-write shared-memory model~\cite{AADGMS93}.   

A process invokes a task with an input value and the task  
returns an output value, so that the inputs 
and the outputs across the processes invoked the task respect the task
specification and every correct process that participates decides (gets an output).
More precisely, a \emph{task} is defined through a set $\I$ 
of input vectors (one input value for each 
process), a set $\O$ of output vectors (one output value for each process),
and a total relation $\Delta:\I\mapsto 2^{\O}$ that associates each input vector 
with a set of possible output vectors.
An input $\bot$ denotes a \emph{non-participating} process and
an output value $\bot$ denotes an \emph{undecided} process.

For vectors $S$ and $S'$ in $\I$ (resp., $\O$),
we write $S\geq S'$ if $S'$ is obtained from $S$
by replacing some entries with $\bot$.  
We assume that if $\I$ (resp., $\O$) contains a
vector $S$, then $\I$ (resp., $\O$)  also contains any vector $S'$
such that $S\geq S'$.  
%
We stipulate that if $(I,O)\in\Delta$, then
	(1) for all $i$, if $I[i]=\bot $, then $O[i]=\bot$,
	(2) for each $O'$, such that $O\geq O'$,  $(I,O')\in\Delta$ and,
	(3) for each $I'$ such that $I'\geq I$, there exists some $O'$ 
such that $O'\geq O$ for all $i$, if $I'[i]\neq\bot $, then $O'[i]\neq\bot$, and $(I',O')$ in $\Delta$.

An algorithm \emph{solves a task $T=(\I,\O,\Delta)$ in a wait-free manner} 
if it ensures that in every execution in which
processes start with an input vector $I\in\I$, every correct process decides,
and the set of decided values, taken together with the processes
taking these decisions, form a vector $O\in\O$ (where positions of non-decided
processes are assigned $\bot$) such that $(I,O)\in\Delta$.    

\myparagraph{The task of $k$-set consensus.}
In the task of \emph{$k$-set consensus},  
	input values are in a set of values $V$  ($|V|\geq k+1$),
	output values are also in $V$,
	and for each input vector $I$ and output vector $O$, $(I,O) \in\Delta$ if 
	the set  of non-$\bot$ values in $O$ is 
	a subset of values in $I$ of size at most $k$.
The special case of $1$-set consensus is called \emph{consensus}~\cite{FLP85}.
More generally, \emph{$(\ell,k)$-set-consensus objects} ($k\leq
\ell)$
allow arbitrary subset of $\ell$ processes to solve $k$-set
consensus.


%
Note that $k$-set consensus is an example of a \emph{colorless} task
(also known as a \emph{convergence} task~\cite{BGLR01}):  
processes are free to use each others' input and output values, so the task can be defined in terms of 
input and output \emph{sets} instead of vectors.
Formally,
let $\val(U)$ denote the set of non-$\bot$ values in a vector $U$.
In a colorless task, for all input vectors $I$ and $I'$
and all output vectors $O$ and $O'$, 
such that $(I,O)\in\Delta$, $\val(I)\subseteq\val(I')$ and
$\val(O')\subseteq\val(O)$, 
we have $(I',O')\in\Delta$.
To solve a colorless task, it is sufficient to find an algorithm 
that allows just one process to decide.
Indeed, if such an algorithm exists, we can simply convert it 
into an algorithm that allows every correct process to decide:
every process simply applies the decision function 
to the observed state of any process that has decided and adopts the decision.

In contrast, $(\ell,k)$-set consensus is not colorless in a system of 
$n>\ell$ processes, as it does not always allow a process to adopt the
decision of another process:
e.g., if a process does not belong to a set $S$ of $\ell$ processes, it cannot
provide outputs for $j$-set consensus for $S$.


\myparagraph{Simulation tools.}
%
%
An execution of a given algorithm $\A$ by the processes $p_1,\ldots p_n$  can be
\emph{simulated} by a set of  \emph{simulator} processes $s_1,\ldots,s_{\ell}$
(or, simply, simulators) that run a distributed  algorithm ``mimicking'' the steps of $\A$ in a
\emph{consistent} way.
Informally, for every execution $E_s$ of the simulation algorithm, there exists an
execution $E$ of $\A$ by $p_1,\ldots,p_n$ such
that the sequence of states simulated for every process $p_i$ in $E_s$ 
is observed by $p_i$  in $E$.


%
A basic building block of our simulations is an \emph{agreement}
protocol~\cite{BG93b,BGLR01} that can be seen as a safe part of consensus.
It exports one operation
$\id{propose}()$ taking $v \in V$ as a parameter and returning $w\in
V$, where $V$ is a (possibly infinite) \emph{value set}. 
When a process $p_i$ invokes $\id{propose}(v)$ we say that $p_i$
\emph{proposes} $v$, and when the invocation returns $v'$ we say that 
$p_i$ \emph{decides on} $v'$.
Agreement ensures four properties: 
\begin{enumerate}
\item[(i)] every decided value has been previously
proposed,
\item[(ii)] no two processes decide on different values, and 
\item[(iii)] if every participating process takes enough steps  then
eventually every correct participating process decides.
\end{enumerate}
Here a process is called participating if it took at least one step in
the computation.
In fact, the agreement protocol in~\cite{BG93b,BGLR01} ensures that 
if every participating process takes  at least three
shared memory steps then
eventually every correct participating process decides.
If a participating process fails in the middle of an agreement
protocol, then no process is guaranteed to return. 

A generalized version of the agreement protocol,
\emph{$\ell$-agreement}~\cite{BG93-TR,ChaR96}, relaxes safety properties of agreement
but improves liveness. Formally, in addition to (i) above, $\ell$-agreement ensures:
\begin{enumerate}
\item[(ii$'$)] at most $\ell$ different values can be decided, and 
\item[(iii$'$)] every correct participating process is guaranteed to decide, unless $\ell$ or more participating
processes do not take enough steps.
\end{enumerate}
Clearly, the agreement protocol we defined above is $1$-agreement.
An $\ell$-agreement protocol with a proof (only sketched
in~\cite{BG93-TR,ChaR96}) can be found in Reiners' thesis~\cite{Rei96}.
For completeness, given that the thesis is not easy to find, we
present the proof in Appendix~\ref{app:lsa}.

\section{A characterization of set-consensus collections}
\label{sec:char}

In this section, we introduce the notion of \emph{agreement level} for a given
set-consensus collection $C$ and a given system size.
Then we show that the metrics captures the power of
$C$ for solving set consensus.  
Then we show how to efficiently compute the agreement level of a
given collection.

\subsection{Agreement levels of $C$}

Consider a model in which processes can communicate via
an atomic-snapshot memory and set-consensus objects from a collection
$C$. 
For brevity, we represent $C$ as a set 
$\{(\ell_0,j_0)$, $(\ell_1,j_1)$, $\ldots$, $(\ell_m,j_m)\}$ such that  
for each $i=0,\ldots,m$, the task of $(\ell_i,j_i)$-set consensus can
be solved ($\ell_i\geq j_i$). 


By convention, we assume that $(\ell_0,j_0)=(1,1)$ is always contained
in a collection $C$: $(1,1)$-set consensus
is trivially solvable.
Note that $(\ell,j)$-set consensus also solves 
$(\ell',j')$-set consensus  for all $\ell'\leq\ell$ and $j'\geq j$.
Thus, without loss of generality, we can assume that the sequence
$(\ell_1,j_1),\ldots,(\ell_m,j_m)$ is
monotonically increasing:  $\ell_0<\ell_1$ and for all $i=1,\ldots,m-1$,
$\ell_i<\ell_{i+1}$ and $j_i<j_{i+1}$.
(since we required that
$(\ell_0,j_0)=(1,1)$, there can be two elements of
the type $(-,1)$).
In particular, for all $n$,  $C$ contains at most $n$ elements $(\ell,j)$ such
that $\ell\leq n$.

%
%


\begin{definition}\label{def:setcon}[Agreement level] Let
  $C=\{(\ell_0,j_0),(\ell_1,j_1),\ldots,(\ell_m,j_m)\}$
 be a collection  of set-consensus objects.
The \emph{agreement level for  $n$ processes of $C$}, denoted
$\AL_n^C$, is defined as:

\begin{itemize} 

\item $\min \sum_i j_i x_i$

\item under the constraints: $\sum_i
\ell_i x_i  \geq n$, $x_0,\ldots,x_m\in \{0,\ldots,n\}$ 

\end{itemize}  

\end{definition}
One can also interpret $\AL_n^C$ as the lowest $k$ for which there exists a 
 \emph{multiset} $S=\{(t_1,s_1)$, $\ldots$, $(t_p,s_p)\}$ of elements in $C$
 such that $\sum_{i}s_i = k$ and $\sum_{i} t_i \geq
 n$.\footnote{Note that assuming that $(1,1)\in C$
 implies $\AL_n^C\leq n$.} 

\subsection{Agreement levels and set consensus}

We now can define a simple criterion to determine whether the model
defined by $C$ can solve
$(n,k)$-set consensus. The criterion is \emph{sufficient}, i.e.,
every model equipped with $C$ that satisfies the criterion solves
$(n,k)$-set consensus, and
\emph{necessary}, i.e., every model equipped with $C$ that solves $(n,k)$-set
consensus satisfies the criterion.
 
\begin{theorem}\label{th:setcon}
$(n,k)$-set consensus can be solved using read-write registers and
 any number of objects taken in a set-consensus collection $C$ if and only if $\AL_n^C\leq k$. 
\end{theorem}
\begin{proof}
Suppose that $\AL_n^C\leq k$.
Thus, there exists a multiset $S=\{(t_1,s_1),\ldots,(t_p,s_p)\}$ of elements in $C$
such that $\sum_{i}s_i \leq k$ and $\sum_{i} t_i \geq n$. 
We show how $n$ processes can solve $k$-set consensus using $S$.  
Every $p_i$, $i=1,\ldots,n$, is assigned to the element $(t_j,s_j)\in S$ such that 
$\sum_{\ell=1,\ldots,j-1} t_{\ell} < i \leq  \sum_{\ell=1,\ldots,j}t_{\ell}$, invokes
the assigned object of $(t_j,s_j)$-set consensus with its input and
returns the corresponding output.
Since  $\sum_{i}s_i \leq k$, the total number of outputs does not
exceed $k$.

Now suppose that $C$ can be used to solve $(k,n)$-set consensus and
let $A$ be the corresponding algorithm.
By contradiction, suppose that no multiset $S$ satisfying the
conditions above exists for $C$. 
Thus, for any multiset $\{(t_1,s_1),\ldots,(t_p,s_p)\}$ of
elements in $C$ such that $\sum_{i}s_i \leq k$,  we have $\sum_{i}
t_i<n$. 

We show that we can then use a simulation of $A$ to solve
$(k+1,k)$-set consensus
using only read-write memory, contradicting the classical
impossibility result~\cite{BG93b,HS99,SZ00}.
The simulation we describe below is an extension of the \emph{BG
simulation}~\cite{BG93b,BGLR01},
inspired by the algorithms described in~\cite{BG93-TR,ChaR96}.  

\vspace{2pt}\noindent \textit{Simulation.}
Let $q_1,\ldots,q_{k+1}$ be a set of $k+1$ simulator processes
communicating via an atomic-snapshot memory.
In its position in the snapshot memory, every simulator $q_i$ maintains its estimate of the
current simulated state of every simulated process in
$\{p_1,\ldots,p_n\}$.

Note that the state of each $p_{\ell}$ (in algorithm $A$)
unambiguously determines the next step that 
$p_{\ell}$ is going to take in the
simulation, which can be an update operation, a snapshot operation, or an access to
a $(t,s)$-set-consensus object.
Since  each update operation by $p_{\ell}$ is implicitly simulated
by registering the last simulated state of $p_{\ell}$ in the shared memory, the simulators only need
to explicitly simulate snapshot operations and accesses to
set-consensus objects.  

We associate each state of $p_{\ell}$ (assuming distinct local states)
with a distinct agreement
protocol (cf. Section~\ref{sec:model}), depending on the next step $p_{\ell}$ is going to take in
that state:

\begin{itemize}

\item For a snapshot operation, we use one instance of the agreement ($1$-agreement)
algorithm. 

\item For an access to a $(t,s)$-set-consensus object, we use one
  instance of
  $s$-agreement 
  and one instance of $1$-agreement.

\end{itemize}

The initial state of each simulated process is associated with a
$1$-agreement protocol.    

The simulation proceeds in asynchronous rounds.
In each round, a simulator $q_i$ picks up the next
simulated process $p_{\ell}$ in a round-robin fashion.  
To simulate a step of  $p_{\ell}$, $q_i$ takes a snapshot of the
memory and computes $p_{\ell}$'s latest simulated state by choosing the latest
simulated state of $p_{\ell}$ found in the snapshot.    

If $p_{\ell}$ is in the initial state, $q_i$ invokes the 
agreement protocol ($1$-agreement) to compute the input of $p_{\ell}$ in
the simulated run, using its input value (for $k$-set consensus)
as a proposed value.  Otherwise, $q_i$ invokes the corresponding agreement protocol:

\begin{itemize}

\item To simulate a snapshot operation, $q_i$ invokes the
  corresponding $1$-agreement protocol,  proposing the just read
  simulated system state (the vector of  the latest simulated
  states of processes $p_1,\ldots,p_n$) as the outcome of the simulated snapshot. 

Recall that a simulator that has started but not finished the
$1$-agreement protocol for a given snapshot operation may block the simulated process
forever. However,  since the faulty simulator may be involved in at most one
agreement protocol at a time, it can block at most one simulated process.  
  
\item To simulate an access of a $(t,s)$-set-consensus object, the
  simulator invokes the corresponding $s$-agreement protocol
  proposing $p_{\ell}$'s input value for this object (according to the simulated
  state) as the decided value.  

Recall that an $s$-agreement protocol may block forever if $s$ or more
processes fail in the middle of its execution.
Thus, when it is used to simulate an access to $(t,s)$-set consensus,
failures of $s$ or more simulators may block $t$ simulated processes.

Also, recall that $s$-agreement may return different values to
different simulators (as long as there are at most $s$ of them).   
To ensure that the outcome of  each of the $t$ simulated processes accessing the
$(t,s)$-set-consensus object is determined consistently by different
simulators, the outcome of the simulated step is then agreed upon using
$1$-agreement.  

\end{itemize}

If an agreement protocol for process $p_{\ell}$ blocks, simulator $q_i$ proceeds to the next
non-blocked simulated process in the round-robin order. 
If the corresponding agreement protocol terminates, the simulator
updates the atomic-snapshot memory with its estimation of the
simulated states of $p_1,\ldots,p_n$, where the new state of $p_i$ is
based on the outcome of the agreement.        

\vspace{2pt}\noindent \textit{Correctness.}
The use of $1$-agreement protocols for both kinds of simulated
operations implies that every step is simulated
consistently, i.e., the simulators agree on the next simulated state of each
process in $\{p_1,\ldots,p_{n}\}$.

The proposal to each of these agreement protocols is either
the recently taken snapshot of the simulated system state
(in case a snapshot operation is simulated) or the value that
the simulated process must propose based on its state (in case
an access to   a $(t,s)$-set-consensus object is simulated).
The initial state of each simulated process is an (agreed upon)
input value of a participating simulator.    

Each simulated snapshot is computed based on the most recent simulated
states of $p_1,\ldots,p_{n}$ contained in the snapshot taken by the
simulator ``winning'' the corresponding $1$-agreement. 
The use of  $s$-agreement in simulating accesses to a $(t,s)$-set-consensus object ensures that the simulated accesses return at most $s$
proposed values.  
Thus, starting from the initial states of the simulated processes, we
inductively derive that all states that appear in the simulated run are \emph{compliant} with a
run $E$ of $A$: in $E$, each process $p_i$ goes through the sequence
of states that are agreed upon for $p_i$ in the simulation.

\vspace{2pt}\noindent \textit{Progress.}
It remains to show that at least one process in
$\{p_1,\ldots,p_{n}\}$ \emph{makes progress} in the simulated run, 
assuming that at least one of the $k+1$ simulators is correct. 
Consider any simulated run.
We show that in this run, at least one of the 
simulated processes takes sufficiently many 
simulated steps (for producing an output
for $k$-set consensus).

A simulated process may stop making progress only if an agreement
protocol used for simulating its step blocks, which may happen only if
a certain number of simulators stopped taking steps in
the middle of the protocol.

Suppose that at most $k$ simulators are faulty.
Given that a faulty simulator can block at most one agreement
protocol, we can identify the set of distinct agreement protocols
$A_1\ldots, A_p$ that are blocked in our run,
and for all $j=1,\ldots,p$, let $A_i$ be $s_i$-agreement.  
We also identify $p$ subsets of $k$ faulty simulators of sizes
$s_1,\ldots,s_p$, where $s_i$ is the number of simulators that block
$A_i$.

For each $i=1,\ldots,p$, let $t_i$ denote the number of simulated processes that
are blocked because of $A_i$.
If $A_i$ is an instance of $1$-agreement ($s_i=1$) used to simulate a snapshot
operation, to agree on the input of a given process,
or to agree on the output of a set-consensus object at a
given process,
then we set $t_i=1$ (only the corresponding simulated process can be
blocked).
Otherwise, $A_i$ is an instance of $s_i$-agreement
used to simulate an access to some $(f_i,s_i)$-set consensus,
and we set $t_i=f_i$ (up to $f_i$ processes accessing the
$(f_i,s_i)$-set-consensus object can be blocked).

Since there are at most $k$ faulty simulators, we get a multiset
$\{(t_1,s_1),\ldots,(t_p,s_p)\}$ of elements in $C$ such that 
$\sum_{i}s_i \leq k$.
But then, by our contradiction hypothesis, we have $\sum_{i} t_i<n$,
i.e., the total number of blocked simulated processes is less than $n$. 
Thus, at least one of the $n$ processes $p_1,\ldots,p_n$ makes progress in the simulated run and
eventually decides. 
Assuming that the first simulator to witness a decision in the
simulated run writes it in the shared memory, we derive that  
every correct simulator eventually reads some decided value and decides. 

Since all these values are coming from a run of an algorithm solving
$(n,k)$-set consensus, there are at most $k$ distinct decided values.
Each of the decided values is an input of some simulator. 
Thus, $k+1$ simulators solve $k$-set-consensus using reads and writes---a contradiction. 
\end{proof}



\subsection{Computing the power of set-consensus collections}
Having characterized the power of a collection $C$ to
solve set-consensus, we are now faced with the question of how to
compute this power. 
                                                                                      
By Theorem~\ref{th:setcon}, determining the best level of agreement
that can be achieved by $C=\{\ell_0,j_0)$, $\ldots$, $(\ell_m,j_m)\}$ in a system of $n$ processes is equivalent
to finding $\min \sum_i j_i x_i$, under the constraints: $\sum_i
\ell_i x_i  \geq n$,  $x_0,x_1,\ldots,x_m\in \{0,\ldots,n\}$.  
This can be viewed as a variation of the Knapsack optimization problem~\cite{knapsack}, 
where we aim at minimizing  the total weight of a set of items from $C$
put in a knapsack, while maintaing a predefined minimal total value of the
knapsack content.\footnote{The classical Knapsack optimization problem
  consists in maximizing the total value, while maintaining the total weight within a given bound.}    
Here each $j_i$ serves as the ``weight''
of an element in $C$, i.e., how much disagreement it may incur, and each
$\ell_i$ serves as its ``value'', i.e., how many processes it is able to
synchronize. 
We use this observation to derive an algorithm to 
compute $\AL_n^C$ in $O(n^2)$ steps.
                                                                                                                                                                                                                                                                                                                                        
Recall that $C$ is represented as a monotonically
increasing sequence $(\ell_0,j_0)$, $\ldots$, $(\ell_m,j_m)$.
First we \emph{complete} $C$  for the fixed system size $n$:
for each $i=1,\ldots,m$, such that $j_i<n$ , we insert elements
$(\max(j_{i}+1,\ell_{i-1}+1),j_{i})$,$(\max(j_{i}+1,\ell_{i-1}+1)+1,j_i)$,
$\ldots$, $(\min(\ell_i,n)-1,j_i)$, $(\min(\ell_i,n),j_i)$.
For example, the completion of $C=\{(1,1),(3,2), (10,6)\}$ for $n=11$ would be 
$\{(1,1)$,  $(3,2)$, $(7,6)$, $(8,6)$, $(9,6)$,
$(10,6)\}$.  
Notice that since $(\ell_0,j_0)$, $\ldots$, $(\ell_m,j_m)$ is
monotonically increasing, such a  completion can be performed in $O(n)$ steps,
and the resulting sequence 
is also monotonically increasing.

As a result of the completion, for every $r=1,\ldots,n$, and each    
element of the kind $(\ell,j)\in C$ such that $\ell>r$ and $j<r$
we have a new element $(r,j)$.
As we will see below, this allows us to compute $\AL_r^C $ in $O(r^2)$ steps.     

%
We observe that for all $r=1,\ldots,n$, $\AL_r^C= \min_{\ell_i\leq r} (j_i+\AL_{r-\ell_i}^C)$.
Indeed, for all $(\ell_i,j_i)$ such that $\ell_i\leq r$, it must hold
that  $j_i+\AL_{r-\ell_i}^C \geq \AL_{r}$, otherwise,
$(\ell_i,j_i)$ plus the multiset $(t_1,s_1),\ldots,(t_p,s_p)$ of elements in $C$ that reaches
$\AL_{r-\ell_i}$ would give $j_i+\sum_v  s_v < \AL_r^C$ and
$\ell_i+\sum_v t_v \geq \ell_i + r -\ell_i = r$, contradicting the
definition of  $\AL_r^C$.
Further, since $C$ is complete, for each multiset
$(t_1,s_1),\ldots,(t_p,s_p)$ in $C$
reaching $\AL_{r}^C$, we can construct a multiset
$(\min(t_1,r),t_1),\ldots,(\min(t_p,r),s)$ in $C$  (each set-consensus object
in the multiset is defined for at most $r$ processes) that also reaches
$\AL_{r}^C$.
Hence, $\AL_r^C= \min_{\ell_i\leq r} (j_i+\AL_{r-\ell_i}^C)$

Thus, we can use the following simple iterative algorithm (a variant of
a solution to the  Knapsack optimization problem based on dynamic
programming) to compute $\AL_n^C$ in $O(n^2)$ steps:

\begin{tabbing}
 bbb\=bb\=bb\=bb\=bb\=bb \kill
 \>$\AL_0^C=0$;\\
 \> \textbf{for} $r=1,\ldots,n$ \textbf{do} $\AL_r^C= \min_{\ell_i\leq r} (j_i+\AL_{r-\ell_i}^C)$.  
\end{tabbing}

In each iteration $r=1,\ldots,n$ of the algorithm above,  
we perform at most $r$ checks, which gives us $O(n^2)$ total complexity.

We can also consider a related notion of \emph{$j$-set-consensus
  number} of $C$, denoted $\SCN_j^C$ and  
defined as the maximal number of processes that can achieve $j$-set
consensus using $C$ and read-write registers: $\SCN_j^C=\max_{\AL_n^C\leq j} n$. 
This is a natural generalization of Herlihy's consensus
power~\cite{Her91}.  
Note that the problem of computing $\SCN_j^C$ is the classical
Knapsack optimization problem, and using a variation of  the algorithm above
we can do it in $O(j|C|)$ steps (see, e.g., \cite[Chap. 5]{knapsack}).     

\section{An adaptive algorithm: reaching optimal agreement}
\label{sec:adapt}

%

Theorem~\ref{th:setcon} implies that, for every fixed $n$,  there exists an $\AL_n^C$-set-agreement algorithm
$\S\T_n^C$ ($\S\T$ for \emph{static}) using $C$. 
We show that these algorithms can be used in an
\emph{adaptive} manner, so that for each set of participating
processes, the best possible level of agreement can be
achieved.

To understand the difficulty of finding such an adaptive algorithm, 
consider $C=\{(1,1)$, $(13,5)$, $(20,9)\}$. 
For selected sizes of participating sets $m$, the table in Figure~\ref{tab:example} gives 
$\AL_m^C$,  and lists the elements of $C$ used in the corresponding $\S\T_m^C$.

%
%


\begin{figure}[thp]
 {\small
\begin{center}
\begin{multicols}{2}
\begin{tabular}{|c|c|p{4cm}|}
\hline
&&\\[-3mm]
$m$&$\AL_m^C$ & $\S\T_m^C$ \\[1mm]

\hline
\hline
1&1&(1,1) \\
2&2&(1,1) (1,1)\\
3 &3&(1,1) (1,1) (1,1)\\
4&4&(1,1) (1,1) (1,1) (1,1)\\
5  to 13&5&(13,5)\\
14 &6&(13,5) (1,1)\\
15 &7&(13,5) (1,1) (1,1)\\
\hline
\end{tabular}
\columnbreak

\hspace{3mm}
\begin{tabular}{|c|c|p{4cm}|}
  \hline
&&\\[-3mm]
$m$&$\AL_m^C$ & $\S\T_m^C$ \\[1mm]

\hline
\hline
16&8&(13,5)(1,1)(1,1)(1,1) \\
17  & 9&(13,5)(1,1)(1,1)(1,1)(1,1) or (9,20)\\
18 to 20  &9&(20,9)\\
21 &10&(20,9)(20,9) or (13,5)(13,5)\\
22 to 26&10& (13,5) (13,5)\\
...&&\\
\hline
\end{tabular}
\end{multicols}
\end{center}
\caption{Selecting elements in $C=\{(1,1)$, $(13,5)$, $(20,9)\}$ 
to solve $\AL_m^C$-set consensus.}
\label{tab:example}
}
\end{figure}

If we have $16$ processes, $\S\T_{16}^C$ uses one instance of
$(13,5)$-set consensus  and three instances of $(1,1)$-set consensus  to achieve  $\AL_{16}^C =9$.
But if two new  processes arrive
we need $(20,9)$-set consensus to achieve $\AL_{18}^C  =9$.
Interestingly,  to achieve $\AL_{22}^C  =10$, we should abandon
$(20,9)$-set consensus
and use  two instances of $(13,5)$-set consensus instead. 
In other words, we cannot simply add a set-consensus instance of the
species we used before to account for the arrival of new processes.
Instead, we have to introduce a new species.

We present a wait-free adaptive algorithm that ensures that if
the set of participating processes is of size $m$, then at most $\AL_m^C$ distinct
input values can be output.  We call such an algorithm \emph{optimally
  adaptive for $C$}.

\vspace{1em}
\begin{figure*}[hbt!]
\hrule
{
\small
\setcounter{linenumber}{0}
\begin{tabbing}
bbb\=bbb\=bbb\=bbb\=bbb\=bbb\=bbb\=bbb\=bbb\=bbb\=bbb\=bbb\=  \kill\\
\emph{Shared objects:} \\ [2mm]
\>\> $R$: \texttt{snapshot object, storing pairs}  $(\textit{value},  \textit{level})$,  \texttt{ initialized to}  $(\bot,\bot)$ \\[3mm]


\emph{Local variables for process $p\in\Pi$:}\\[1mm]
\>$r[1,\ldots,n]$:  \texttt{ array of pairs}  $(\textit{value},  \textit{level})$ \\
\> $\textit{prop},v$: $value$  \\ 
\> $\textit{parts}, P \in 2^{\Pi}$ \\ 
\>$index$: integer\\[3mm]

\emph{Code for process $p\in\Pi$ with proposal $v_p$:}\\[1mm]

\nnll\label{update1}\>$R.\textit{update}(p, (v_p,0))$ \\[1mm]
\nnll\label{scan1}\>$r[1,\ldots,n]=R.\textit{snapshot}()$\\[1mm]
\nnll\label{s1}\>$P$=set of processes $q$ such that $r[q]\neq (\bot,\bot)$\\[1mm]


\nnll\label{nb:repeat}\>\textbf{repeat}\\
\nnll\label{s6}\>\>$\textit{parts} := P$\\[1mm]
\nnll\label{s7}\>\>$\textit{rank} :=$ the rank of $p$ in $\textit{parts}$\\[1mm]
\nnll\label{s9}\>\> $k := $ be the greatest integer such that $(-,k)$ is in $r$\\[1mm]
\nnll\label{s10}\>\> $v :=$ be any value such that $(v,k)$ is in $r$\\[1mm]
\nnll\label{s8}\>\>$\textit{prop} := {\S\T}_{|\textit{parts}|}^C$
with value $v$ at position  $\textit{rank}$\\[1mm]
\nnll\label{update2}\>\>$R.\textit{update}(p, (\textit{prop},{|\textit{parts}|}))$\\[1mm]
\nnll\label{scan2}\>\>$r[1,\ldots,n]=R.\textit{snapshot}()$\\[1mm]
\nnll\label{s5}\>\>$P$=set of processes $q$ such that $r[q]\neq (\bot,\bot)$\\[1mm]
\nnll\label{nb:until}\>\textbf{until} $\textit{parts}=P$\\[1mm]
\nnll\label{nb:exit} \> \textbf{return} $prop$

\end{tabbing}}
\hrule
\caption{An optimally  adaptive set-consensus algorithm.} 
\label{fig:adapt}
\end{figure*}

The algorithm is presented in Figure~\ref{fig:adapt}. 
The idea is the following: periodically, every process $p$ writes its
current value (initially, its input), together
with the number of processes it has seen participating so far
(initially, $0$)  in the shared memory, and takes a snapshot to get the current set $P$  of participating
processes and their inputs.

Process $p$ then computes its \emph{rank} in $P$ and adopts the value $v$
from a process announcing the largest participating set. 
The chosen input is then proposed to an instance of algorithm $\S\T_{|P|}^C$, where 
$p$ behaves as the process at position \textit{rank} and proposes value $v$.
More precisely, $\S\T_{|P|}^C$ is treated as an algorithm for processes
$q_1,\ldots,q_{|P|}$ and, thus, $p$ runs the code of process
$q_{\textit{rank}}$ in the algorithm with input value $v$.       

Note that since the set is derived from an atomic snapshot of the memory, the
notion of the largest participating set is well-defined: the snapshots
of the same size are identical. 
Therefore, at most $|P|$ processes participate in $\S\T_{|P|}^C$ and
each of these $|P|$ processes can only participate at a distinct
position corresponding to its rank in $P$.
As a result, every correct process invoking $\S\T_{|P|}^C$ will
eventually get an output of the ``best'' set-consensus algorithm for $P$.   

When the participating set $P$ observed by $p$
does not change in two consecutive
iterations, $p$ terminates with its current value.

\begin{theorem}\label{th:adapt}
Let  $\cal C$ be a set-consensus collection,  $n$ be an integer.
The algorithm in Figure~\ref{fig:adapt} is optimally adaptive for
$\cal C$ in a system of $n$ processes.
\end{theorem}

\begin{proof}  
We show first that every correct
process eventually returns a value, and any returned value is a
proposed one.

Let $p$ and $q$ take snapshots (Lines~\ref{scan1} or~\ref{scan2}) in
that order, and let  $P_p$  and $P_q$ be, respectively, the returned
participating sets. 
We observe first that $P_p\subseteq P_q$. 
Indeed, each position in the snapshot object $R$ is initialized to
$(\bot,\bot)$. Once, $p$ updates $R[p]$ with its value and
participation level, the position remains non-$\bot$ forever.
Thus, if $q$ takes its snapshot of $R$ after $p$, then $P_p$, the set
of processes whose positions are non-$\bot$ in the resulting
vector, is a subset of $P_q$.    

Therefore, the sets $P$ and $\textit{parts}$ evaluated by $p$ in Line~\ref{scan2}
are non-decreasing with time. 
Since $\S\T_{|\textit{parts}|}^C$ is wait-free, the only reason for a
correct process $p$ not to return is to find that
$\textit{parts}\subsetneq P$ in Line~\ref{nb:until} infinitely often, i.e.,
both $P$ and  $\textit{parts}$ grow indefinitely. 
But the two sets are bounded by the set $\Pi$ of all processes---a contradiction.

Furthermore, every returned value is a value decided in an instance of
$\S\T_{|\textit{parts}|}^C$. 
But every value proposed to algorithm  $\S\T_{|\textit{parts}|}^C$ was
previously read in a non-$(\bot,\bot)$ position of $R$, which can only
contain an input value of some process.

Hence, every correct process eventually returns a value, and any returned value is a
proposed one.

Now consider a run of the algorithm in Figure~\ref{fig:adapt} in which $m$ processes participate. 
We say that a process $p$ \emph{returns at level $t$} in this run if
it outputs (in Line~\ref{nb:exit}) the value $\textit{prop}$ returned
by the preceding invocation of $\S\T_t^C$ (in Line~\ref{s8}).
By the algorithm, if $p$ returns at level $t$, then the set $\textit{parts}$ of
processes it witnessed participating is of size $t$.

Let $\ell$ be the smallest level ($1\leq \ell\leq n$) at which some
process returns, and let $O_{\ell}$ be the set of values ever written
in $R$ at level $\ell$, i.e., all values $v$, such that $(v,\ell)$
appears in $R$.

We show first that for all $\ell'>\ell$, if $R$ contains $(v',\ell')$, then $v'\in
O_{\ell}$.
 
By contradiction, suppose that some process $q$ is the first process to write a value $(v',\ell')$ (in
Line~\ref{update2}), such
that $\ell'>\ell$ and $v'\not\in O_{\ell}$, in $R$.
Thus, the immediately preceding snapshot taken by $q$ before this write  (in
Lines~\ref{scan1} or~\ref{scan2}) witnessed a participating set of
size $\ell'$. Hence, the snapshot of $q$ occurs after the last snapshot (of size
$\ell<\ell'$) taken by any process $p$ that returned at level $\ell$.
But immediately before taking its last snapshot, every such process
$p$ has written $(v,\ell)$ in $R$ (Line~\ref{update2}) for some $v\in O_{\ell}$.
Thus $q$ must see  $(v,\ell)$ in its snapshot of size $\ell'$ and,
since, by the assumption, the snapshot contains no values written at
levels higher than $\ell$, $q$ must adopt some value written at level
$\ell$ (Lines~\ref{s9} and~\ref{s10}). Thus, $v'\in O_{\ell}$---a contradiction. 

Thus, every returned value must appear in
$O_{\ell}$, where $\ell$ is the smallest level ($1\leq \ell \leq n$)
at which some process returns. 
Now we show that 
$|O_{\ell}|\leq \AL_m^C$, recall that $m$ is the number of
participating processes. 

Indeed, since all values that appear in $O_{\ell}$ were previously returned by
the algorithm $\S\T_{\ell}^C$ (Line~\ref{s8}) and, as we observed
earlier, the algorithm is used by at most $\ell$ processes, each
choosing  a unique position based on its rank in the corresponding
snapshot of size $\ell$, there can be at most $\AL_{\ell}^C$ such
values. Since at most $m$ processes participate in the considered run,
we have $\ell\leq m$, and, thus, $\AL_{\ell}^C\leq \AL_m^C$. 

Hence, in a run with
participating set of size $m$, $|O_{\ell}|\leq \AL_m^C$ and, thus, at
most $\AL_m^C$ values can be returned by the algorithm.
Thus, we indeed have an optimally adaptive set-consensus algorithm
using $C$. 
\end{proof}

\myparagraph{On unbounded concurrency.} Our definitions of the agreement
level and the set-consensus number of a set-consensus collection  are
independent of the size of the
system: they are defined with respect to a given participation level. 
Our adaptive algorithm does account for the system size, as it uses
atomic snapshots. 
But by employing the atomic-snapshot algorithms for
unbounded-concurrency models described in~\cite{GMT01}, we can easily
extend our adaptive solution to these models too.

\section{Related work}
\label{sec:related}
Our algorithm computing the power of a set-consensus collection in
$O(n^2)$ steps (for a system of $n$ processes) is inspired by  
the dynamic programming solution to the Knapsack optimization problem described,
e.g., in~\cite[Chap. 5]{knapsack}. 
 
Herlihy~\cite{Her91} introduced the notion of \emph{consensus number}
of a given object type, i.e., the
maximum number of processes that can solve consensus using instances
of the type and read-write registers. 
It has been shown that $n$-process consensus objects have consensus
power $n$. However, the corresponding consensus hierarchy is in
general not robust, i.e., there exist object types, each of consensus number $1$
which, combined together, can be used to solve $2$-process
consensus~\cite{LH00}.
Besides objects of the same consensus number $m$ may not be equivalent
in a system of more than $m$ processes~\cite{AEG16}.  

Borowsky and Gafni~\cite{BG93-TR}, and then Chaudhuri and Reiners~\cite{ChaR96,Rei96}
independently explored the power of having
multiple instances of $(\ell,j)$-set-consensus
objects in a system of $n$ processes with respect to solving set
consensus, which is a special case of the question considered in this
paper. The characterization of~\cite{BG93-TR,ChaR96,Rei96} is established by a
generalized BG simulation~\cite{BG93b,BGLR01} by Borowsky and Gafni,
where instead of $1$-agreement protocol, a more general
$j$-agreement protocol is used. Our results employ this
agreement protocol to show a more general result.  

Gafni and Koutsoupias~\cite{GK99-undecidable} and Herlihy and Rajsbaum~\cite{HR97} showed that wait-free
solvability of tasks for $3$ or more processes using registers is an undecidable question. 
We show that in a special case of solving set
consensus using a set-consensus collection, the question is decidable.
Moreover, we give an explicit polynomial algorithm for
computing the power of a set-consensus collection.  

\section{Concluding remarks}
\label{sec:conc}

We hope that this work will be a step towards proving a more general
conjecture that
our set-consensus numbers capture precisely the computing power of
any ``natural'' shared-memory model. 
An indication that the conjecture is true
is that set-consensus objects are, in a precise sense, \emph{universal}
(generalizing the consensus universality~\cite{Her91}): using
$(n,k)$-set-consensus objects, $n$ processes can implement $k$
independent sequential state machines so that at least one of them is able to make
progress, i.e., to execute infinitely many commands~\cite{GG11-univ}.   
Popular restrictions of the runs of the wait-free model, such as
\emph{adversaries}~\cite{DFGT11} and \emph{failure
  detectors}~\cite{CT96,CHT96}, were successfully characterized via
their power for solving set consensus~\cite{GK10,GK11,DFGK15}.  
Also, it can be inferred from the recent result by Afek et
al.~\cite{AEG16} that, like consensus,
set-consensus objects can express precisely certain deterministic objects~\cite{AEG16}.
We therefore believe that the power of a large class of ``natural''
models (determined by restrictions mentioned above) can be captured by
their ability to solve set consensus.
This class must exclude models in which ``in between''
objects, like \emph{Weak Symmetry-Breaking}~\cite{HS99,GRH06}, are
used: such models, as we believe, cannot be expressed as a restriction of
the runs of the wait-free model, and are therefore not ``natural''.


  \def\noopsort#1{} \def\No{\kern-.25em\lower.2ex\hbox{\char'27}}
  \def\no#1{\relax} \def\http#1{{\\{\small\tt
  http://www-litp.ibp.fr:80/{$\sim$}#1}}}

\appendix

\section{An $\ell$-agreement algorithm}
\label{app:lsa}

The algorithm (presented in Figure~\ref{fig:sa}) uses  
two \emph{atomic snapshot} objects $A$ and $B$, initialized with $\bot$'s.
A process writes its input in $A$
(line~\ref{line:sa:updateA}) and takes 
a snapshot of $A$ (line~\ref{line:sa:scanA}).
Then the process writes the outcome of the snapshot in $B$ (line~\ref{line:sa:updateB})  
and keeps taking snapshots of $B$ until
it finds that at most $\ell-1$ participating ( i.e., having written their  values
in $A$) processes that have not finished the protocol, i.e., have not written
their values in $B$ (Lines~\ref{line:sa:repeat}-\ref{line:sa:until}).  
Finally, the process returns the smallest value (we assume that the
value set is ordered) in the smallest-size non-$\bot$ snapshot found
in $B$ (containing the smallest number of non-$\bot$ values). 
(Recall that all snapshot outcomes are related by containment, so
there indeed exists such a smallest snapshot.)  
  
\begin{figure}[tbp]
\hrule \vspace{1mm}
 {\footnotesize
\begin{tabbing}
 bbb\=bb\=bb\=bb\=bb\=bb\=bb\=bb \=  \kill
\emph{Shared objects:} \\
\> $A$, $B$: \texttt{snapshot objects, initially} $\bot$\\
\\
\textit{propose}$(v)$\\
\nnll\label{line:sa:updateA}\> $A.\id{update(v)}$\\
\nnll\label{line:sa:scanA}\> $U:=A.\id{snapshot()}$\\
\nnll\label{line:sa:updateB}\> $B.\id{update(U)}$\\
\nnll\label{line:sa:repeat}\>  {\bf repeat} \\
\nnll\label{line:sa:scanB}\>\> $W:=B.\id{snapshot()}$\\
\nnll\label{line:sa:compute}\>\> $X:= \{j| (U[j]\neq\bot)\;\wedge\;(W[j]=\bot)$\}\\
\nnll\label{line:sa:until}\>  {\bf until } $|X|\leq\ell-1$\\
\nnll\label{line:sa:minscan}\> $\id{S}:=$ the smallest-size
set of non-$\bot$ values contained in
$\{\; W[j];\;j=1,\ldots,n,\; W[j]\neq\bot\}$\\
\nnll\label{line:ca:minvalue}\> $\id{return}$ $\min(S)$ 
\end{tabbing}
\hrule 
}
\caption{The $\ell$-agreement algorithm}
\label{fig:sa}
\end{figure}

\begin{theorem}\label{th:sa}
The algorithm in Figure~\ref{fig:sa} implements $\ell$-agreement.
\end{theorem}
\begin{proof}
The validity property (i) is immediate: only the identifier of a
participating process can be
found in a snapshot object.  
The termination property (iii)$'$ of $\ell$-agreement is immediate: if at most
$\ell-1 $ processes that have executed line~\ref{line:sa:updateA} fail to execute
line~\ref{line:sa:updateB}, 
then the exit condition of the repeat-until clause in
line~\ref{line:sa:until} eventually holds and every correct
participating process
terminates.

Suppose, by contradiction, that at least $\ell+1$ different values are
returned by the algorithm.  
Thus, at least $\ell+1$ distinct snapshots were written in $B$ by
$\ell+1$ processes. 
Let $L$ be the set of processes that have written the $\ell$ smallest
snapshots  in $B$ in the run. 
The set is well-defined as all snapshots taken in $A$ are related by
containment. 
We are going to establish a contradiction by showing that every process must return the smallest value
in one of the snapshots written by the processes in $L$ and, thus, at
most $\ell$ distinct inputs will be produced.

Let $p_i$ be any process that completed line~\ref{line:sa:updateB} by
writing the result of its snapshot of $A$ in $B$.
Let  $U$ be the set of processes that $p_i$ witnessed in $A$ and, thus,
wrote to its position in $B$ in line~\ref{line:sa:updateB}. 


If $p_i\in L$, i.e., $U$ is one of the $\ell$ smallest snapshots ever
written in $B$, then $p_i$ will return the value of the smallest process in $U$ or
a smaller snapshot written by some process in $L$. 
If $p_i\notin L$, then $U$ contains all $\ell$ distinct snapshots
written by the processes in $L$.
Since each process in $L$ is included in the snapshot it has written
in $B$, we derive that $L\subseteq U$.   
Since $p_i$ returns a value only if all but at most $\ell-1$ processes it witnessed
participating have written their snapshots in $B$, at least one
snapshot written by a process in $L$ is read by $p_i$ in $B$. 
Thus, $p_i$ outputs the value of the smallest process in the snapshot written by a
process in $L$---a contradiction.

Thus, at most $\ell$ distinct values can be output and (ii)$'$ is satisfied. 
\end{proof}

\end{document}